\documentclass[1p,final]{elsarticle}
\usepackage{amsfonts,color,morefloats,pslatex}
\usepackage{amssymb,amsthm, amsmath,latexsym}
\usepackage{hyperref}
\usepackage{cleveref}

\allowdisplaybreaks[4]

\newtheorem{thm}{Theorem}[section]

\newtheorem{rem}[thm]{Remark}
\newtheorem{exa}[thm]{Example}
\newtheorem{defn}[thm]{Definition}

\numberwithin{equation}{section}
\newcommand{\lcm}{\operatorname{lcm}}

\begin{document}

\begin{frontmatter}



\title{Minimum Distance of New Generalizations of the Punctured Binary Reed-Muller Codes}

\author[label1]{Liqin Hu\corref{mycorrespondingauthor}}
\ead{huliqin@hdu.edu.cn}
\cortext[mycorrespondingauthor]{Corresponding author, Liqin Hu }
\address[label1]{School of Cyberspace, Hangzhou dianzi University , Hangzhou,310018, China}

\author[label2]{Keqin Feng}
\ead{kfeng@math.tsinghua.edu.cn}
\address[label2]{Department of mathematical science, Tsinghua University, Beijing, 100084, China}

\tnotetext[mytitlenote]{The paper is supported by National Natural Science Foundation of China (No. 11471178, 11571107, 61602144)}
%%\thanks{The paper is supported by National Natural Science Foundation of China (No. 11471178, 11571107, 61602144).}

\begin{abstract}
Motivated by applications in combinatorial design theory and constructing LCD codes, C. Ding et al \cite{DLX} introduced cyclic codes $\mho(q,m,h)$ and $\bar\mho(q,m,h)$ over $\mathbb{F}_q$ as new generalization and version of the punctured binary Reed-Muller codes. In this paper, we show several new results on minimum distance of $\mho(q,m,h)$ and $\bar\mho(q,m,h)$ which are generalization or improvement of previous results given in \cite{DLX}.
\end{abstract}


\end{frontmatter}

\section{Introduction}
In 1954, Reed (\cite{R}) and Muller (\cite{M}) constructed independently a kind of binary linear codes (Reed-Muller codes). The punctured binary Reed-Muller codes are cyclic and have been generalized into ones over arbitrary finite fields $\mathbb{F}_q$ (see \cite{DGM, KLP} and others). Such codes and their variants have applications not only in error correcting, but also in secret sharing, data storage systems (locally testable or locally decodable codes) and computational complexity  theory. Recently, C. Ding et al \cite{DLX} present new generalization $\mho(q,m,h)$ and version $\overline{\mho}(q,m,h)$ of the punctured binary Reed-Muller codes motivated by their applications in combinatorial designs and constructing LCD codes (linear code with complement dual).

Let $q$ be a power of a prime number, $m\ge2$, $\alpha$ be a primitive element of the finite field $\mathbb{F}_{q^m}$, which means $\mathbb{F}_{q^m}^*=<\alpha>$. Each nonzero element in $\mathbb{F}_{q^m}$ can be expressed uniquely by $\alpha^a$ with $0\le a\le n-1$ and $n=q^m-1$. For any integer $a$ with $0\le a\le n-1$, let
$$a=a_0+a_1q+\cdots+a_{m-1}q^{m-1}, 0\le a_i\le q-1,$$
be the $q$-adic expansion of $a$. The Hamming $q$-weight $wt_q(a)$ of $a$ is defined by the Hamming weight of vector $(a_0,a_1,\cdots,a_{m-1})\in \mathbb{Z}_q^m$. Namely,
$$wt_q(a)=\sharp\{i:0\le i\le m-1, a_i\ne0\}.$$

\begin{defn}
Let $q$ be a power of a prime number, $m\ge 2$, $1\le h\le m-1$, $n=q^m-1$ and $\mathbb{F}_{q^m}^*=<\alpha>$. The cyclic codes $\mho(q,m,h)$ and $\overline{\mho}(q,m,h)$ over $\mathbb{F}_q$ are ideals of the ring $\mathbb{F}_q[x]/(x^n-1)$ with the generating polynomial
$$g_{q,m,h}(x)=\prod_{\substack{a=1\\wt_q(a)\le h}}^{n-1}(x-\alpha^a)\in \mathbb{F}_q[x]$$
and
$$\overline{g}_{q,m,h}(x)=(x-1)\lcm(g_{q,m,h}(x),\widehat g_{q,m,h}(x))$$
respectively, where $\widehat g_{q,m,h}(x)$ is the reciprocal polynomial of $g_{q,m,h}(x)$.

Namely, let
$$I(q,m,h)=\{1\le a\le n-1: wt_q(a)\le h\}$$
$$I\_(q,m,h)=\{n-a: a\in I(q,m,h)\}.$$
Then the set of zeros of $g_{q,m,h}(x)$ and $\overline g_{q,m,h}(x)$ are
$$Z(q,m,h)=\{\alpha^a:a\in I(q,m,h)\}$$
and
$$\overline Z(q,m,h)=Z(q,m,h)\cup Z\_(q,m,h)\cup \{1\}$$
respectively, where $Z\_(q,m,h)=\{\alpha^{-a}:a\in I(q,m,h)\}=\{\alpha^a:a\in I\_(q,m,h)\}$.
\end{defn}

The length of the codes $\mho(q,m,h)$ and $\overline \mho(q,m,h)$ is $n=q^m-1$. Let $k(q,m,h)$ and $\overline k(q,m,h)$ be the dimension of $\mho(q,m,h)$ and $\overline \mho(q,m,h)$ over $\mathbb{F}_q$. Then $$k(q,m,h)=n-\deg(g_{q,m,h}(x))=n-\sum_{i=1}^{h}(q-1)^i\dbinom{n}{i}.$$
On the other hand, if $h\le [\frac{m-1}{2}]$, then $I(q,m,h)\cap I\_{ }(q,m,h)=\emptyset$, $\overline g_{q,m,h}(x)=(x-1)g_{q,m,h}(x)\widehat{g}_{q,m,h}(x)$ and $\overline k(q,m,h)=n-1-2\sum_{i=1}^h (q-1)^i\dbinom{n}{i}$. The dimension $\overline k(q,m,h)$ has also determined for $h=[\frac{m+1}{2}]$ (\cite{DLX}, Theorem 26 and 27). In this paper, we focus on the minimum (Hamming) distance $d(q,m,h)$ and $\overline d(q,m,h)$ of the codes $\mho(q,m,h)$ and $\overline\mho(q,m,h)$. The following results have been proved in \cite{DLX}.

\begin{thm}Let $q$ be a power of a prime number, $m\ge2$, $n=q^m-1$ and $1\le h\le m-1$. Then
\begin{description}
  \item[(1)] $\frac{q^{h+1}-1}{q-1}\le d(q,m,h)\le 2q^h-1$ (\cite{DLX}, Theorem 3). Particularly, $d(2,m,h)=2^{h+1}-1$.
  \item[(2)] $\overline d(q,m,h)\ge \frac{2(q^{h+1}-1)}{q-1}$ for $h\le [\frac{m+1}2]$ (\cite{DLX}, Theorem 25-27).
  \item[(3)] $d(3,m,1)=4$ (reaches the lower bound in \cite{BEW}) and $\mho(3,m,1)$ for all $m\ge2$ are distance-optimal (by sphere-packing bound) (\cite{DLX}, Corollary 4).
  \item[(4)] $d(q,m,m-1)=\frac{q^m-1}{q-1}$ (reaches the lower bound in \cite{BEW}) (\cite{DLX}, Theorem 6).
\end{description}

\end{thm}

In section 2 of this paper, we proved the following new results.

\begin{description}
  \item[(I)] We provide a sufficient condition for a divisor $e$ of $n=q^m-1$ such that $d(q,m,h)\le e$ and $\overline d(q,m,h)\le2e$ (Theorem 2.1). In many cases, we can find such divisor $e<2q^h-1$ so that the upper bound $2q^h-1$ of $d(q,m,h)$ can be improved and an upper bound of $\overline{d}(q,m,h)$ is presented.

      As one of direct consequence of this general result , we have showed that

  \item[(II)] If $q\ge3$ and $h+1|m$, then $d(q,m,h)=\frac{q^{h+1}-1}{q-1}$ and $\overline d(q,m,h)=\frac{2(q^{h+1})}{q-1}$  (Theorem 2.2). Namely, $d(q,m,h)$ and $\overline d(q,m,h)$ reach their lower bounds given by Theorem 1.2 (1) and (2) if $h+1|m$. C. Ding et al \cite{DLX} raised open problem 1: Is it true that $d(q,m,h)=\frac{q^{h+1}-1}{q-1}$? From Theorem 1.2 we know that this is true if $q=2$ and any $m\ge 2$, $1\le h\le m-1$, $(q,m,h)=(3,m,1)$ or $h+1=m$. Theorem 2.2 provide new evidence in $h+1|m$ case.

      In section 3 we give more specific consideration for $h=1$ case. By Theorem 1.2 we know that for $m\ge2$,
      $$q+1\le d(q,m,1)\le 2q-1 \mbox{ ,  } \overline d(q,m,1)\ge 2(q+1)$$

      And by Theorem 2.2 we have $d(q,m,1)=q+1$ and $\overline d(q,m,1)= 2(q+1)$ if $m$ is even. We get the following new result on upper bound of $\overline d(q,m,1)$.

  \item[(III)] For $m\ge4$, $\overline d(2,m,1)=6$ and the codes $\overline \mho(2,m,1)$ are distance-optimal by sphere-packing bound (Theorem 3.1). $\overline d(3,m,1)\le10$ for all $m\ge2$ (Theorem 3.3). For $q\ge4$, $\overline d(q,m,1)\le2(2q-1)$ for sufficient large $m\ge c(q)$ (Theorem 3.2).

  \item[(IV)] As an application of Theorem 2.1, we provide a simple sufficient condition for a divisor $e$ of $n=q^m-1$ such that $d(q,m,1)\le e$ and $\overline d(q,m,1)\le 2e$ (Theorem 3.4). For $7\le q\le32$, we make a table for such $e<2q-1$ so that the upper bounds $d(q,m,1)\le 2q-1$ and $\overline d(q,m,1)\le 2(2q-1)$ are improved.
\end{description}

In last section, we make conclusion and raise some open problems for further research.

At the end of this section, we remark that C. Ding et al \cite{DLX} also introduced the extended code $\widehat{\mho}(q,m,h)$ of $\mho(q,m,h)$ and show that the minimum distance $\widehat{d}(q,m,h)$ is $d(q,m,h)+1$. Thus any result on $d(q,m,h)$ can be shifted to the one on $\widehat d(q,m,h)$ directly.

\section{General Case}
Let $\mathbb{F}_{q^m}^*=<\alpha>$, $m\ge2$, $1\le h\le m-1$, $n=q^m-1$ and
$$I(q,m,h)=\{a:1\le a\le n-1, 1\le wt_q(a)\le h\}$$
$$I\_(q,m,h)=\{n-a: a\in I(q,m,h)\}$$
$$\overline{I}(q,m,h)=I(q,m,h)\cup I\_(q,m,h)\cup\{0\}.$$
Then the set of zeros of the cyclic codes $\mho(q,m,h)$ and $\overline\mho(q,m,h)$ is $\{\alpha^a:a\in I\_(q,m,h)\}$ and $\{\alpha^a:a\in \overline I(q,m,h)\}$ respectively. For $0\le a,b\le n-1$, we call $a$ and $b$ belong to a same $q$-cyclotomic class, if there exists $l\in\mathbb{Z}$ such that $a\equiv q^lb\pmod n$. The set $\{0,1,\cdots,n-1\}$ is divided into $q$-cyclotomic classes. For any $x\in \mathbb{Z}$, we defined $wt_q(x)=wt_q([x]_n)$, where $[x]_n$ is the least non-negative residue of $x$ modulo $n$ ($0\le [x]_n\le n-1$). If $a=a_0+a_1q+\cdots+a_{m-1}q^{m-1}$ is the $q$-adic
expansion of $a$, $0\le a\le n-1=q^m-2$, then $[aq]_n=a_{m-1}+a_0q+\cdots+a_{m-2}q^{m-1}$, so that $wt_q(aq)=wt_q(a) (=\sharp\{i:0\le i\le m-1,a_i\ne0\}\in \mathbb{Z})$. Thus $I(q,m,h)$ is divided into disjoint $q$-cycloomic classes $I_1, I_2, \cdots, I_r$. For each $I_i$, we take $a_i\in I_i$. The set $R(q,m,h)=\{a_1,a_2,\cdots,a_r\}$ is called a representative system of $I_1, I_2, \cdots, I_r$. Usually we take $a_i$ to be the smallest integer in $I_i$. With these notations we know that
$$\mho(q,m,h)=\{c(x)=\sum_{i=0}^{n-1}c_ix^i\in\mathbb{F}_q[x]/(x^n-1):c(\alpha^a)=0 \mbox{ for all }a\in R(q,m,h)\}$$
$$\overline\mho(q,m,h)=\{c(x)=\sum_{i=0}^{n-1}c_ix^i\in\mathbb{F}_q[x]/(x^n-1):c(1)=c(\alpha^a)=c(\alpha^{-a})=0 \mbox{ for all }a\in R(q,m,h)\}$$

After all these preparation, we present the following general result.

\begin{thm}Let $q$ be a power of a prime number, $m\ge2$, $1\le h\le m-1$, and $n=q^m-1$.  For a divisor $e$ of $n$, if $2\le e<n$ and the following condition (*) holds,
 \begin{description}
   \item[(*)] $e\nmid a$ for all $a\in I(q,m,h)$
 \end{description}
 then for any integer $l\ge 1$, $d(q,ml,h)\le e$ and $\overline d(q,ml,h)\le 2e$.
\end{thm}
\begin{proof}Let $e\ge2$ be a divisor of $n=q^m-1$ satisfying the condition (*). Let
$$M=ml\mbox{ , }n=ef\mbox{ , }N=q^M-1=eF$$
where $F=f\cdot \frac{q^{ml}-1}{q^m-1}\in \mathbb{Z}$. Consider the following polynomial
$$c(x)=\sum_{i=0}^{N-1}c_ix^i=\frac{x^N-1}{x^F-1}\in\mathbb{F}_q[x]/(x^N-1).$$
We claim that $c(x)\in\mho(q,ml,h)$.

 Let $\mathbb{F}_{q^M}^*=<\alpha>$. Then
\begin{eqnarray*}
   & c(x)\in\mho(q,m,h) & \Leftrightarrow c(\alpha)=0\mbox{ for all }a\in I(q,M,h)=\{a: 1\le a\le N-1,1\le wt_q(a)\le h\} \\
   &&\Leftrightarrow \alpha^a\mbox{ is not a zero of }x^F-1\mbox{ for all }a\in I(q,M,h).
\end{eqnarray*}
But
\begin{eqnarray*}
   & \alpha^a\mbox{ is a zero of }x^F-1&\Leftrightarrow \alpha^{aF}=1 \\
   &&\Leftrightarrow N|aF\mbox{ since the order of } \alpha\mbox{ is }N\\
   &&\Leftrightarrow e|a\mbox{ since }N=eF.
\end{eqnarray*}
Therefore,
\begin{equation}\label{2.1}
    c(x)\in\mho(q,m,h)\Leftrightarrow e\nmid a\mbox{ for all }a\in I(q,m,h)
\end{equation}
 For each $a\in I(q,m,h)$, we know that $1\le a\le q^M-2$, $1\le wt_q(a)\le h$ so that $a$ has the following $q$-adic expansion
 $$a=\sum_{i=0}^{M-1}a_iq^i,\mbox{ }(0\le a_i\le q-1).$$
 If there exists $i\ge m$ such that $a_i\ge1$, let $a'=a-a_iq^i+a_iq^{i-m}$. Then $1\le a'<a$, $a'\equiv a\pmod{n}$ and
 $$wt_q(a')\le wt_q(a-a_iq^i)+wt_q(a_iq^i)\le (h-1)+1=h.$$
 After finite step of this procedure, we get an integer $\widetilde{a}$ such that $1\le \widetilde{a}\le n-1=q^m-2$, $\widetilde{a}\equiv a\pmod{n}$ and $wt_q(\widetilde{a})\le h$. Therefore, $\widetilde{a}\in I(q,m,h)$ and $\widetilde{a}\equiv a\pmod e$, since $e|n$. Then we get that $e\nmid a$ if and only if $e\nmid \widetilde{a}$. This implies that if the condition (*) holds, then the right-hand side of (\ref{2.1}) is true, so that $c(x)\in \mho(q,m,h)$. From $N=eF$ and
 $$c(x)=\frac{x^N-1}{x^F-1}=x^{(e-1)F}+x^{(e-2)F}+\cdots+x^F+1,\mbox{ }\deg(c(x))=(e-1)F<N,$$
 we know that $c(x)$ is a non-zero codeword in $\mho(q,M,h)$ and the Hamming weight is $wt_H(c(x))=e$. Therefore $d(q,M,h)\le e$.

 Next we consider
 $$\overline c(x)=c(x)(x-1)=\frac{(x-1)(x^N-1)}{x^F-1}.$$
 If the condition (*) holds we have proved that $\overline c(\alpha^a)=0$ for all $a\in I(q,m,h)$. Moreover, by similar argument, it can be shown that if the condition (*) holds, then $\overline c(\alpha^a)=0$ for all $a\in I\_(q,m,h)=\{-b:1\le b\le n-1, 1\le wt_q(b)\le h\}$ since $e|a$ if and only if $(-e)|a$. Finally, $\overline c(1)=0$. Therefore, $\overline c(\alpha^a)=0$ for all $a\in \overline I(q,m,h)$ which means that $\overline c(x)\in\overline\mho(q,M,h)$. From $\deg (\overline c(x))=(e-1)F+1<N(=eF)$ we know that $\overline c(x)$ is a non-zero codeword in $\overline \mho(q,m,h)$ and the Hamming weight $wt_H(\overline c(x))\le wt_H(c(x))\cdot wt_H(x-1)=2e$. Therefore, $\overline d(q,M,h)\le 2e$. This completes the proof of Theorem 2.1.

\end{proof}

Remark:

\begin{description}
             \item[(1)] Since $d(q,m,h)\ge \frac{q^{h+1}-1}{q-1}$, any divisor $e$ of $n$ satisfying the condition (*) should be at least $\frac{q^{h+1}-1}{q-1}$. On the other hand, if $e\le 2(q^h-1)$, then Theorem 2.1 presents a better upper bound of $d(q,m,h)$ than the one in Theorem 1.2(1).
             \item[(2)] For any integer $a$ with $a\not\equiv0\pmod{n}$, $e\nmid a$ if and only if $e\nmid aq$ since $\gcd(e,q)=1$. This means that if $a$ and $b$ belong to the same $q$-cyclotomic class, then $e\nmid a$ if and only if $e\nmid b$. Let $\{I_1,\cdots,I_r\}$ be the partition of $I(q,m,h)$ with $I_i$ ($0\le i\le r$) being $q$-cyclotomic classes, $R(q,m,h)=\{a_1,\cdots,a_r\}$ is a representative set of this partition. Then  $e\nmid a$ for all $a\in I(q,m,h)$ is equivalent to the following condition:
                 \begin{description}
                   \item[(R)] $e\nmid a$ for all $a\in R(q,m,h)$.
                 \end{description}

                 Moreover, consider $R(q,m,h)$ as a partial order set with respect to the divisibility order $\preceq$: $a\preceq b$ if and only if $a|b$. Let $M(q,m,h)$ be the set of maximal elements of $(R(q,m,h),\preceq)$. It is easy to see that both of the condition (*) and (R) are equivalent to the following condition
                  \begin{description}
                   \item[(M)]$e\nmid a$ for all $a\in M(q,m,h)$.
                   \end{description}
 \end{description}

\begin{exa} Take $q=3$, $h=2$, $m\ge4$, $n=3^m-1$. Theorem 1.2 gives $13=\frac{3^2-1}{3-1}\le d(3,m,2)\le 2\cdot 3^2-1=17$, $\overline d(3,m,2)\ge26$. For $m=4$, we have
\begin{eqnarray*}
R(3,4,2) &=& \{1,2\}\cup\{a+3b: a,b\in\{1,2\}\}\cup\{a+9b: 1\le b\le a\le2\} \\
         &=& \{1,2,4,5,7,8,10,11,20\},
\end{eqnarray*}
$M(3,4,2)=\{7,8,11,20\}$, $n=3^4-1=80$,  $e=16$ is a divisor of $n$, and $e\nmid a$ for all $a\in M(3,4,2)$. By Theorem 2.1, we get $d(3,4,2)\le16$ and $\overline d(3,4,2)\le32$.

For $m=6$, $n=3^6-1=2^3\cdot 7\cdot13$, we have
\begin{eqnarray*}
  R(3,6,2) &=& R(3,4,2)\cup\{a+27b:1\le b\le a\le2\} \\
           &=& \{1,2,4,5,7,8,10,11,20\}\cup\{28,29,58\}
\end{eqnarray*}
$M(3,6,2)=\{8,11,20,28,58\}$. Take $e=13$, by Theorem 2.1, we get $d(3,6,2)=13$ and $\overline d(3,6,2)=26$.
\end{exa}

Ding et al. raised several open problems in \cite{DLX}. One of them is : Is it true that $d(q,m,h)=\frac{q^h-1}{q-1}$ (the lower bound given in Theorem 1.2)? Theorem 1.2 shows that this is true for four cases: $(q,m,h)=(2,m,h)$, $(q,m,h)=(q, \mbox{even }m,1)$, $(q,m,h)=(q,m,m-1)$, and $(q,m,h)=(3,m,1)$.
As an application of Theorem 2.1, we give the following generalization of the first three cases which show more evidence for this open problem.

\begin{thm}Let $q\ge3$ be a power of a prime number, $m\ge 2$, $1\le h\le m-1$. Then for each $l\ge1$, $d(q,(h+1)l,h)=\frac{q^{h+1}-1}{q-1}$ and $\overline d(q,(h+1)l,h)=\frac{2(q^{h+1}-1)}{q-1}$.
\end{thm}
\begin{proof} Take $e=\frac{q^{h+1}-1}{q-1}$, $m=h+1$ in Theorem 2.1. Then $n=ef$, $f=q-1$. For $1\le t\le f-1$,
$$wt_q(te)=wt_q(\sum_{i=0}^htq^i)=h+1>h.$$
Therefore, $e$ is not a divisor of any number $a$ in $I(q,m,h)$ since $wt_q(a)\le h$. By Theorem 2.1 and 1.2 we get $d(q,(h+1)l,h)=e=\frac{q^{h+1}-1}{q-1}$ and $\overline d(q,(h+1)l,h)=2e$.
\end{proof}

\section{$h=1$ case}

In this section we deal with $h=1$ case more precisely. From the first two sections we know that
\begin{description}
  \item[(I)] For $q\ge3$ and $m\ge3$, $q+1\le d(q,m,1)\le 2q-1$, $\overline d(q,m,1)\ge2(q+1)$.
  \item[(II)] For $m\ge4$, $d(2,m,1)=3$ and $\overline d(2,m,1)\ge6$ for $m\ge3$, $d(3,m,1)=4$ and $\overline d(3,m,1)\ge8$.
  \item[(III)] For $q\ge3$ and even number $m\ge4$, $d(q,m,1)=q+1$ and $\overline d(q,m,1)=2(q+1)$.
\end{description}

Now we present an upper bound of $\overline d(q,m,1)$. Firstly we consider $q=2$ case.

\begin{thm}\label{3.1}
For all $m\ge4$, $\overline d(2,m,1)=6$ and the code $\overline\mho(2,m,1)$ is distance optimal.(Remark that $\overline\mho(2,3,1)=\{0\}$).
\end{thm}

\begin{proof}
We know that $\overline d(2,m,1)\ge6$. The parameters of binary code $\overline\mho(2,m,1)$ is $[n,k,d]$ where $n=2^m-1$, $k=dim_{\mathbb{F}_2}\overline\mho(2,m,1)=n-|I\_(2,m,1)|=n-(1+2m)$. If $d(=\overline d(2,m,1))\ge7$, the sphere-packing bound gives
\begin{equation}\label{3.1}
    2(n+1)^2=2^{1+2m}=2^{n-k}\ge\sum_{i=0}^3\dbinom{n}{i}=1+n+\frac12n(n-1)+\frac16n(n-1)(n-2),
\end{equation}
which is equivalent to $n^3-12n^2-19n-6\le0$. Let $f(x)=x^3-12x^2-19x-6$. It can be checked that $f(15)=f'(15)=f''(15)=f'''(15)>0$. This implies that for $m\ge4$, $n=2^m-1\ge15$ and $f(n)>0$ which contradicts to the inequality (\ref{3.1}) given by the sphere-packing bound. Therefore $\overline d(2,m,1)=6$ and the codes $\overline\mho(2,m,1)$ for $m\ge4$ are distance-optimal.

\end{proof}

The sphere-packing can also be used to obtain an upper bound of $\overline d(q,m,1)$ for $q\ge3$.

\begin{thm}\label{3.2}
For any fixed $q\ge3$, there exists a constant $c=c(q)>0$ such that $\overline d(q,m,1)\le2(2q-1)$ for all odd integers $m>c$.
\end{thm}
\begin{proof}The parameters of the cyclic code $\overline\mho(q,m,1)$ is $[n,k,d]_q$, where
$$n=q^m-1\mbox{ , }k=n-|I\_(q,m,1)|=n-(1+2(q-1)m)\mbox{ , }d=\overline d(q,m,1).$$
Suppose that $\overline d(q,m,1)\ge2(2q-1)+1$. The sphere-packing bound gives
\begin{equation}\label{2.3}
    q^{1+2(q-1)m}=q^{n-k}\ge\sum_{i=0}^{2q-1}(q-1)^i\dbinom{n}{i}.
\end{equation}
The last term of the right hand side is
$$(q-1)^{2q-1}\dbinom{q^m-1}{2q-1}=M(q)q^{(2q-1)m}+O(q)q^{(2q-2)m}\mbox{ , }M(q)=\frac{(q-1)^{2q-1}}{(2q-1)!}>0.$$
When $q$ is fixed and $m\rightarrow \infty$, we know that the equality (\ref{2.3}) cannot hold for sufficient large $m$ since $2(q-1)<2q-1$. This completes the proof of Theorem 3.2.
\end{proof}

By more careful estimation, it would be obtained an explicit value of $c(q)$. The case $q=3$ is easy.

\begin{thm}\label{3.3}
For any odd integer $m\ge3$, $\overline d(3,m,1)\le10$.
\end{thm}

\begin{proof}
If $\overline d(3,m,1)\ge11$, the sphere-packing bound gives
\begin{equation}\label{2.4}
3(n+1)^4=3^{4m+1}\le\sum_{\lambda=0}^52^{\lambda}\dbinom{n}{\lambda}\mbox{ (}n=3^m-1\mbox{ )}
\end{equation}
which, by an elementary computation, is equivalent to
\begin{equation}\label{2.5}
4n^5-75n^4-80n^3-390n^2-104n-30\le0.
\end{equation}
But when $m\ge3$, $n\ge26$ and the left-hand side of (\ref{2.5}) is
\begin{eqnarray*}
   && n^5(4-\frac{75}n-\frac{80}{n^2}-\frac{390}{n^3}-\frac{104}{n^4}-\frac{30}{n^5}) \\
   &\ge& n^5(4-\frac{75}{26}-\frac{80}{26^2}-\frac{390}{26^3}-\frac{104}{26^4}-\frac{30}{26^5}) >n^5(4-3.3)>0.
\end{eqnarray*}
Therefore, the equality (\ref{2.4}) does not hold for any odd number $m\ge3$. This completes the proof of Theorem 3.3.
\end{proof}

Next, we show that the upper bound $d(q,m,1)\le 2q-1$ can be improved in many cases by using Theorem 2.1 and a remarkable fact: $R(q,m,1)=\{1,2,\cdots,q-1\}$ is independent of $m$. At the same cases we also present an upper bound of $\overline d(q,m,1)$ which is smaller than $2(2q-1)$. From $d(q,m,1)\ge q+1$ we know that $d(q,m,1)=q+a$, where $1\le a\le q-1$.

\begin{thm}\label{3.4}
Let $q$ be a power of a prime number.
\begin{description}
  \item[(1)] Let $1\le a\le q-2$, $e=q+a$ and $\gcd(a,q)=1$. If the order of $(-a)$ modulo $e$ is an odd integer $l$, then $d(q,\lambda l,1)\le e$ and $\overline d(q,\lambda l,1)\le2e$ for all odd integer $\lambda\ge1$.
  \item[(2)] Let $m$ be an odd positive integer. If $e$ is a divisor of $n=q^m-1$ and $q+1\le e\le 2q-1$, then $d(q,m,1)\le e$ and $\overline d(q,m,1)\le 2e$.
\end{description}
\end{thm}

\begin{proof}\begin{description}
  \item[(1)] By assumption and $q\equiv-a\pmod{e}$ we know that the order of $q$ modulo $e$ is $l$ and then $e|q^l-1$. From $e=q+a>q-1$ we know that $e\nmid x$ for any $x\in R(q,m,1)=\{1,2,\cdots,q-1\}$. Then the conclusion is derived from Theorem 2.1 and Remark (2) after Theorem 2.1.
    \item[(2)] It is a direct consequence of (1). The order  $l$ of $q$ modulo $e$ is a divisor of $m$, therefore $l$ is odd and $m=l\lambda$, $\lambda\in \mathbb{Z}$.
\end{description}
\end{proof}

\begin{rem}\begin{description}
             \item[(1)] For $a=1$, the order of $-1$ modulo $q+1$ is two. We obtain the previous results: for even number $m\ge2$, $d(q,m,1)\le q+1$ and $\overline d(q,m,1)\le2(q+1)$.
             \item[(2)] The following facts in elementary number theory may be helpful to judge if the order $b$ modulo $e$ is odd for $b,e\in\mathbb{Z}$, $e\ge2$ and $\gcd(b,e)=1$. We denote by $O_e(b)$ the order of $b$ modulo $e$. Namely, $O_e(b)$ is the least integer $l\ge1$ such that $b^l\equiv1\pmod{e}$.
                 \begin{description}
                   \item[(F1)] $O_e(b)|\phi(e)$, where $\phi(e)$ is the Euler function defined by
                   $$\phi(e)=\sharp\{i:1\le i\le e,\gcd(i,e)=1\}.$$
                   \item[(F2)] Let $e=p_1^{a_1}\cdots p_s^{a_s}$, where $p_1,\cdots,p_s$ are distinct prime numbers, $a_i\ge1$, $1\le i\le s$. Then $O_e(b)$ is odd if and only if $O_{{p_i}^{a_i}}(b)$ is odd for all $i$ ($1\le i\le s$) since $O_e(b)=\lcm\{O_{{p_i}^{a_i}}(b):1\le i\le s\}$.
                   \item[(F3)] For $e=2^a(a\ge1)$, $\phi(e)=2^{a-1}$. Then $O_e(b)$ is odd if and only if $O_e(b)=1$, which means that $b\equiv1\pmod{e}$. For $e=p^a$, where $p$ is an odd prime and $a\ge1$, $\phi(e)=p^{a-1}(p-1)$. We know that if $p\nmid b$, then
                       $$O_{p^2}(b)=O_p(b)\mbox{ or }O_p(b)p\mbox{ , }O_{p^a}(b)=O_{p^2}(b)\cdot p^{a-2}\mbox{ (for } a\ge3\mbox{ )}.$$
                       Therefore, $O_p(b)$ is odd for $a\ge2$ if and only if $O_p(b)$ is odd.
                   \item[(F4)] For an odd prime number $p$, let
                       $$\phi(p)=p-1=2^aN,$$
                       where $N$ is odd, $a\ge1$. Then for $p\nmid b$,
                       $$O_p(b)\mbox{ is odd }\Leftrightarrow O_p(b)|N\Leftrightarrow b\mbox{ is a }2^a-th\mbox{ power modulo }p.$$
                   Particularly, if $p\equiv3\pmod4$ and $p\nmid b$, $O_p(b)$ is odd if and only if $b$ is a square (quadratic residue) modulo $p$ (which can be determined by the quadratic reciprocity law).
                 \end{description}
                 There exist several elementary criteria to judge if $b$ is a $4$-th and $8$-th  power modulo $p$ for smaller $|b|$ (see \cite{BEW}, Chapter 7).
                 \end{description}
\end{rem}

\begin{exa}
Take $q=25$. We know that for $m\ge2$, $q+1=26\le d(25,m,1)\le 49=2q-1$, $\overline d(25,m,1)\ge52$. Moreover, $\overline d(25,m,1)\le98$ for sufficient large $m$, and for all $\lambda\ge1$,
$$d(25,2\lambda,1)=26\mbox{ , }\overline d(25,2\lambda,1)=52.$$
\end{exa}

Now we use Theorem \ref{3.4}(1). In the following, $\lambda$ denotes any odd positive integer.

For $a=2$, $e=q+a=27$, $O_e(-2)=9$. Therefore $d(25,9\lambda,1)\le27$ and $\overline d(25,9\lambda,1)\le 54$.

For $a=3$, $e=28$, $O_e(-3)=3$. Therefore $d(25,3\lambda,1)\le28$ and $\overline d(25,3\lambda,1)\le 56$.

For $a=4$, $e=29$, $O_e(-4)=7$. Therefore $d(25,7\lambda,1)\le29$ and $\overline d(25,7\lambda,1)\le 58$.

For $a=8$, $e=33$, $O_e(-8)=5$. Therefore $d(25,5\lambda,1)\le33$ and $\overline d(25,5\lambda,1)\le 66$.

For $a=22$, $e=47$, $O_e(-22)=23$. Therefore $d(25,23\lambda,1)\le47$ and $\overline d(25,23\lambda,1)\le 94$.

For $7\le q\le 32$, the following tables presents all integer $a$ such that $2\le a\le q-2$ and $l=O_{q+a}(-a)$ is odd. Theorem \ref{3.4}(1) implies that for all odd $\lambda\ge1$, $d(q,l\lambda,1)\le q+a$ and $\overline d(q,l\lambda,1)\le 2(q+a)$. We also list the general lower bound $q+1$ and upper bound $2q-1$ of $d(q,m,1)$.

\begin{center}
Table I
\begin{tabular}{|c|c|c|c|cc|ccccc|cc|c|ccccc|}
    \hline
    $q$         &7  &9  &11  &13 &    &16  &   &   &   &    &19  &    &23  &25  &   &   &   &  \\ \hline
    $a$         &2  &2  &3  &5 &10  &3  &5  &7  &9  &13  &4  &8   &6  &2  &3  &4  &8  &22 \\
    $l=O_e(-a)$ &3  &5  &3  &3 &11  &9  &3  &11 &5  &7   &11 &3   &7  &9  &3  &7  &5  &23  \\
    $e=q+a$     &9  &11 &14 &18 &23 &19 &21 &23 &25 &29  &23 &27  &29 &27 &28 &29 &33 &47 \\
    $q+1$       &8  &10 &12 &14 &   &17  &   &   &   &   &20 &    &24 &26  &   &   &   &  \\
    $2q-1$      &13 &17 &21 &25 &  &31  &   &   &   &   & 37 &    &45 & 49  &   &   &   &  \\
    \hline
  \end{tabular}
\end{center}

\begin{center}
Table II

\begin{tabular}{|c|cc|ccc|cccc|cc|}
    \hline
    $q$         &27 &   &29  &   &   &31 &    &   &   &32 & \\ \hline
    $a$         &19 &20 &17 &20 &23 &2 &12  &14 &15 &15 &17 \\
    $l=O_e(-a)$ &11 &23 &11 &7  &3  &5 &21  &3  &11 &23 &21  \\
    $e=q+a$     &46 &47 &46 &49 &52 &33 &43 &45 &46 &47 &49 \\
    $q+1$       &28 &   &30 &   &   &32 &   &  &   &33 &  \\
    $2q-1$      &53 &   &57 &   &   &61 &   &  &   &63 &  \\
    \hline
  \end{tabular}
\end{center}

 \section{conclusion}
 This paper present several new results on the minimum distance $d(q,m,h)$ and $\overline d(q,m,h)$ of the cyclic codes $\mho(q,m,h)$ and $\overline\mho(q,m,h)$ introduced in \cite{DLX}. The results are summarized in Introduction (Section 1, (I)-(IV)). Particularly,
 \begin{enumerate}
   \item We find that the codes $\overline\mho(2,m,1)$ for all $m\ge4$ with parameters $[2^m-1,2^m-2m-2,6]$ are distance-optimal.
   \item We generalize a result given in \cite{DLX} to prove that $d(q,(h+1)\lambda,h)=\frac{q^{h+1}-1}{q-1}$ and $\overline d(q,(h+1)\lambda,h)=\frac{2(q^{h+1}-1)}{q-1}$ for any $\lambda\ge1$ and $q\ge3$.
   \item By using the sphere-packing bound, we show that for any fixed $q$, $\overline d(q,m,1)\le2(2q-1)$ for sufficient large $m(\ge c(q))$. And for $q=3$, $\overline d(3,m,1)\le10$ for all $m\ge2$.
   \item By using codeword $c(x)=\frac{x^n-1}{x^f-1}$ in $\mho(q,m,h)$ and $\overline{c}(x)=(x-1)c(x)$ in $\overline\mho(q,m,h)$, we get $d(q,m,h)\le e$ and $\overline d(q,m,h)\le 2e$ for some factor $e\ge1$ of $n=q^m-1$ ($f=n/e$). In many cases these upper bounds are better than $d(q,m,h)\le 2q-1$ and $\overline d(q,m,h)\le2(2q-1)$.
 \end{enumerate}

 Ding et al. \cite{DLX} raised several open problems on codes $\mho(q,m,h)$ and $\overline\mho(q,m,h)$. Focus on their minimum distance, we add the following problems.
 \begin{description}
   \item[(1)] All examples show that $\overline d(q,m,h)\le 2d(q,m,h)$ and $\overline d(q,m,h)\le 2(2q^h-1)$. Are these always true? Remark that $d(q,m,h)\le2q^h-1$. Thus $\overline d(q,m,h)\le2d(q,m,h)$  implies $\overline d(q,m,h)\le2(2q^h-1)$.
   \item[(2)] All examples show that $\overline d(q,m,h)=\frac{2(q^{h+1}-1)}{q-1}$ provided $1\le h\le[\frac{m-1}2]$ for $q\ge3$ and $1\le h\le[\frac{m-2}2]$ for $q=2$. Is this always true?
   \item[(3)] Find new method to improve the upper bound of $d(q,m,h)$ and/ or $\overline d(q,m,h)$. Find non-zero codewords $c(x)\in\mho(q,m,h)$ ($\overline c(x)\in\overline \mho(q,m,h)$) having small Hamming weight and with other form than $\frac{x^n-1}{x^f-1}$ ($\frac{(x-1)(x^n-1)}{x^f-1}$), then we get $d(q,m,h)\le wt_H(c(x))$ ($\overline d(q,m,h)\le wt_H(\overline c(x))$).
 \end{description}

%\ifCLASSOPTIONcaptionsoff
 \newpage
%\fi

\end{document}